\documentclass[12pt]{elsarticle}
\usepackage[utf8]{inputenc}
\usepackage{graphicx}
\usepackage{subfigure}

\usepackage{amssymb}
\usepackage{amsthm}
\usepackage{url}
\usepackage{amsmath}
\usepackage{subfigure}
\usepackage{textcomp}
\usepackage{xcolor}
\graphicspath{{figures/}}
\newtheorem{prop}{Proposition}

\title{Computational universality of fungal sandpile automata}
\author[1,2]{Eric Goles}
\author[1]{Michail-Antisthenis Tsompanas}
\author[1]{Andrew Adamatzky}
\author[3]{Martin Tegelaar}
\author[3]{Han A. B. Wosten}
\author[1,4]{Genaro J. Mart{\'i}nez}

\address[1]{Unconventional Computing Laboratory, University of the West of England, Bristol, UK}
\address[2]{Facultad de Ingenier\'{i}a y Ciencias, Universidad Adolfo Ib\'{a}\~{n}ez, Santiago, Chile}
\address[3]{Microbiology Department, University of Utrecht, Utrecht, The Netherlands}
\address[4]{Superior School of Computer Science, National Polytechnic Institute, Mexico}

\date{}

\begin{document}

\begin{frontmatter}

\begin{abstract}
Hyphae within the mycelia of the ascomycetous fungi are compartmentalised by septa. Each septum has a pore that allows for inter-compartmental and  inter-hyphal streaming of cytosol and  even organelles. The compartments, however, have special organelles, Woronin bodies, that can plug the pores. When the pores are blocked, no flow of cytoplasm takes place. Inspired by the controllable compartmentalisation within the mycelium of the ascomycetous fungi we designed two-dimensional fungal automata. A fungal automaton is a cellular automaton where communication between neighbouring cells can be blocked on demand. We demonstrate computational universality of the fungal automata by implementing sandpile cellular automata circuits there. We reduce the Monotone Circuit Value Problem  to  the  Fungal  Automaton  Prediction  Problem. We construct families of wires, cross-overs and  gates to prove that the fungal automata are P-complete. 
\end{abstract}


\end{frontmatter}


\section{Introduction}

Fungi are ubiquitous organisms that are present in all ecological niches. They can grow as single cells but can also form mycelium networks covering up to 10~km$^2$ of forest soil~\cite{smith1992fungus, ferguson2003coarse}. Fungi can sense what humans sense and more, including tactile stimulation~\cite{jaffe2002thigmo,kung2005possible}, pH~\cite{van2002arbuscular}, metals~\cite{fomina2000negative}, chemicals~\cite{howitz2008xenohormesis}, light~\cite{purschwitz2006seeing} and gravity~\cite{moore1991perception}. Fungi exhibit a rich spectrum of electrical activity patterns~\cite{slayman1976action,olsson1995action,adamatzkyspiking}, which can be tuned by external stimulation. On studying electrical responses of fungi to stimulation~\cite{adamatzkyspiking} we proposed design and experimental implementation of fungal computers~\cite{adamatzky2018towards}. Further numerical experiments demonstrated that it is possible to compute Boolean functions with spikes of electrical activity propagating on mycelium networks~\cite{adamatzky2020boolean}. At this early stage of developing a fungal computer architecture we need to establish a strong formal background reflecting several alternative ways of computing with fungi. 

In \cite{fungalautomata} we introduced one-dimensional fungal automata, based on a composition of two elementary cellular automaton functions. We studied the automata space-time complexity and discovered a range of local events essential for a future computing device on a single hypha. In present paper we aim to demonstrate a computational universality of fungal automata. To do this we modify state transition rules of  sand pile, or chip firing, automata~\cite{dhar1990self,goles1992sand,christensen1991dynamical,bitar1992parallel,goles1991sand} to allow a control for moving of sand grains, or chips, between neighbouring cells. The local control of the interactions between cells is inspired by a control of cytosol flow control in fungal hyphae~\cite{moore1962fine,lew2005mass,bleichrodt2012hyphal,bleichrodt2015switching,tegelaar2020subpopulations}. Then we used developed tools of sand pile automata universality~\cite{goles1996sand,goles1997universality,chessa1999universality,moore1999computational,gajardo2006crossing,martin2013goles} to show that functionally complete sets of Boolean gates can be realised in in the fungal automata.

The paper is structured as follows. We define two dimensional fungal automata in  Sect.~\ref{Two dimension Fungal Automata}. We represent the fungal automata as sand pile, or chip firing, automata in Sect.~\ref{The chip firing automata}. Section~\ref{Computational Complexity notions} places fungal automata in the context of computational complexity. The P-completeness of the fungal automata is demonstrated via implementation of Boolean circuits in Sect.~\ref{Computational Complexity of the Fungal Automata}. Alternative versions of the Woronin body states updates are presented in Sect.~\ref{Other words of automaton updates}. Future developments and feasibility of practical implementation are discussed in Sect.~\ref{Discussion}.

\section{Two dimensional Fungal Automata}
\label{Two dimension Fungal Automata}

\begin{figure}[!tb]
    \centering
    \subfigure[]{
    \includegraphics[scale=0.5]{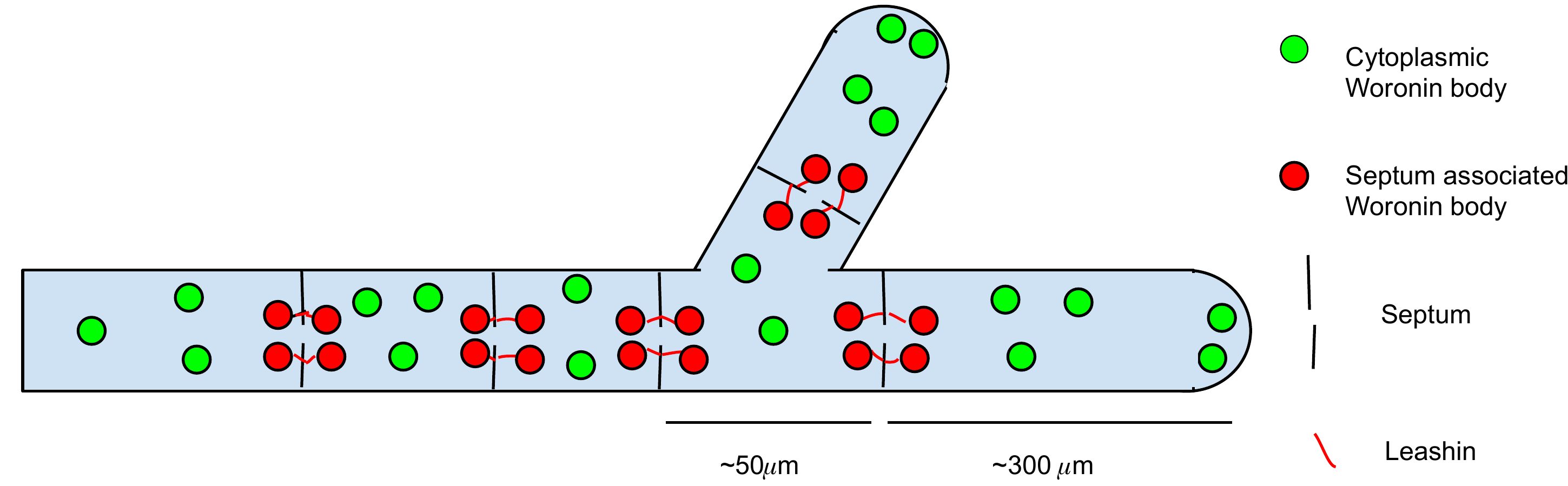}}\\
    \subfigure[]{
    \includegraphics[scale=0.3]{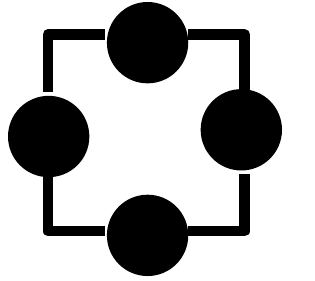}}
    \subfigure[]{
    \includegraphics[scale=0.3]{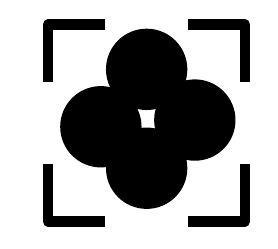}}
\subfigure[]{
    \includegraphics[scale=0.3]{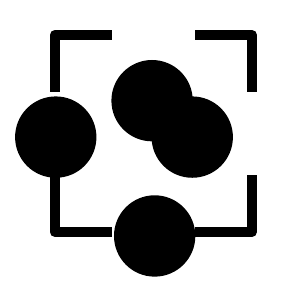}}
    \subfigure[]{
    \includegraphics[scale=0.3]{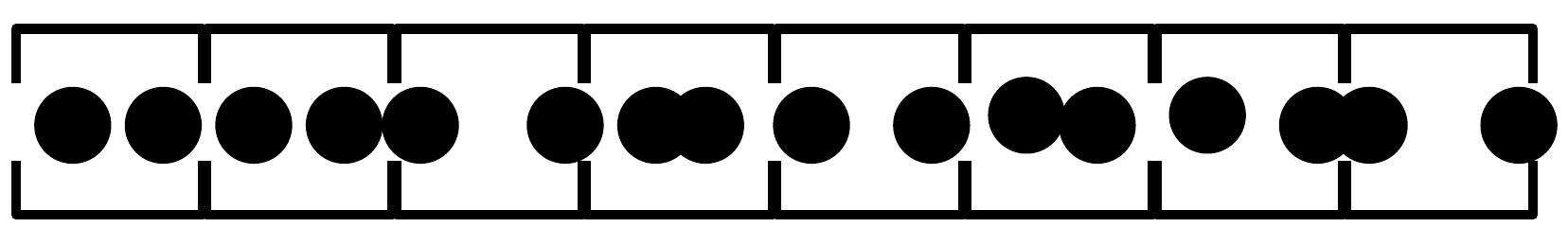}
    }
    \subfigure[]{
    \includegraphics[scale=0.3]{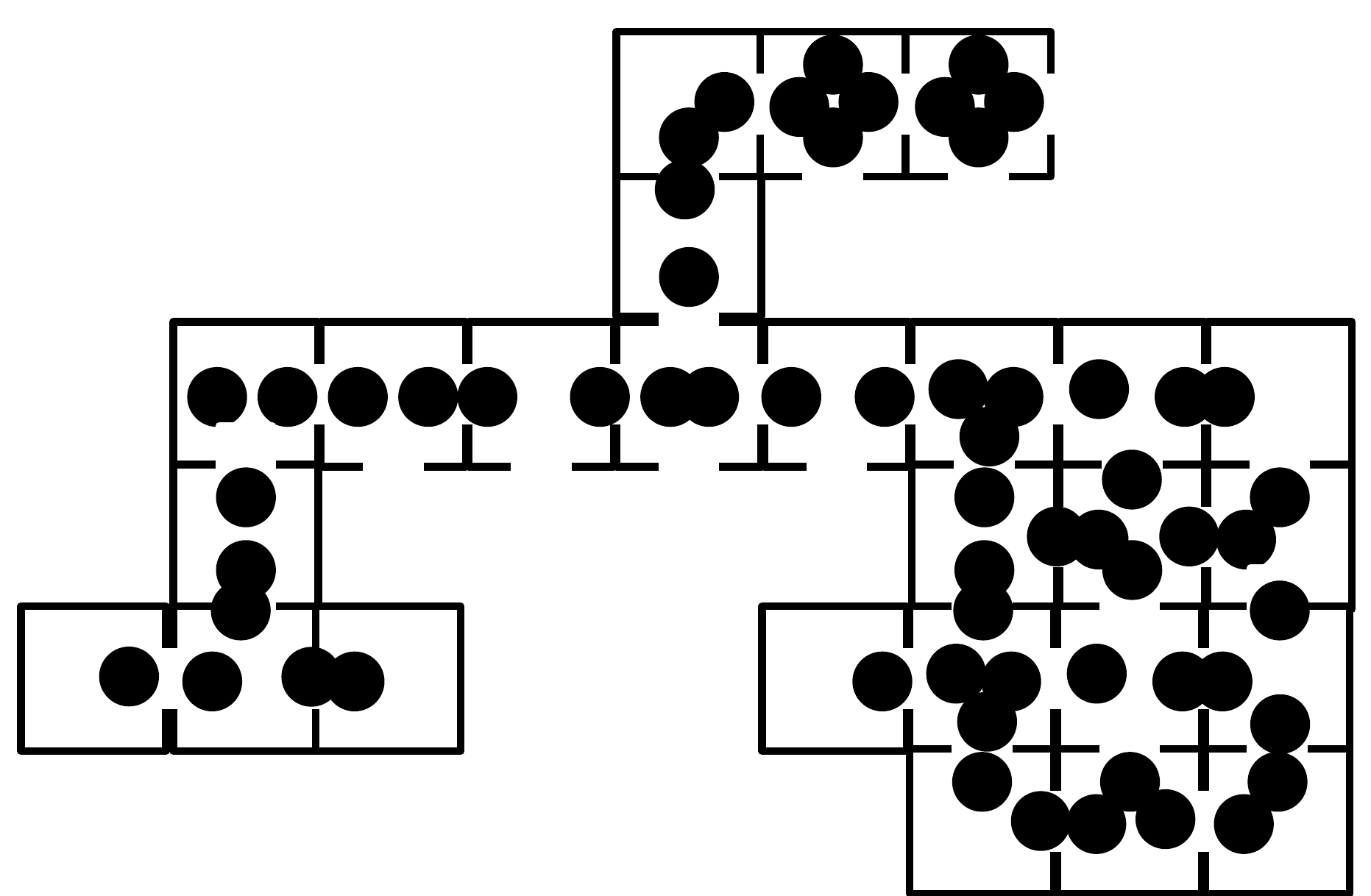}}
    \caption{Fungal automata. (a)~Biological scheme. (bcdef)~Abstract schemes. 
    (b)~All pores are closed, 
    (c)~All pores are open, 
    (d)~North and  East pores are open, 
    (e)~one-dimensional automaton,
    (f)~an example of an arbitrary architecture of fungal automata.
    }
    \label{fig:schemes}
\end{figure}


 Filamentous fungi of the phylum Ascomycota have porous septa that allow for cytoplasmic streaming throughout hyphae and the mycelium~\cite{moore1962fine,lew2005mass}. The pores of damaged hyphae will be plugged by a peroxisome-derived organelle to prevent bleeding of cytoplasm into the environment~\cite{tenney2000hex,reichle1965multiperforate,tey2005polarized,jedd2000new}. These Woronin bodies can also plug septa of intact hyphae~\cite{bleichrodt2012hyphal,steinberg2017woronin}. The septal pore occlusion in these hyphae can be triggered by septal ageing and stress conditions~\cite{bleichrodt2012hyphal,bleichrodt2015selective,tegelaar2020subpopulations}.
 
 A scheme of the mycelium with Woronin bodies is shown in Fig.~\ref{fig:schemes}. An apical compartment has one neighbouring sub-apical compartment, while a sub-apical compartment has a neighbouring compartment at both ends. Because compartments can also branch, they can have one or more additional neighbouring compartments. Thus, the compartment with pores and Woronin bodies is a elementary unit of fungal automata (Fig.~\ref{fig:schemes}bcd). From these compartments one can assemble quasi-one-dimensional (Fig.~\ref{fig:schemes}e) and two-dimensional (Fig.~\ref{fig:schemes}f) structures. 
 
 In this context, let us consider a cellular automaton in the two dimensional grid  $\mathbb{Z}^2$ with the von Neumann neighbourhood, with set of states $Q$  and  a global function $F$. Each cell of the grid has four sides that could be open or closed. An open side means that the information (the state) of both cells is shared. If not, when the side is closed, both sites mutually ignore each other. When every side is open we have the usual cellular automata model~\cite{wolfram1994cellular}. On the other hand, if sides are open or closed in some random or periodical way we get for the same local functions different dynamic behaviours. In this paper we will consider only ``uniform'' ways to open-close the sides. Actually at a given step to open every vertical side (every column of the grid) or every horizontal side, rows of the array. So the fungal automata model becomes a tuple 
 $FA=\langle \mathbb{Z}^2, \mathbf{Q}, V, F, w \rangle$, 
 where $V$ is the von Neumann neighbourhood,  $w$ is a finite word on the alphabet ${H,V}$ (horizontal, vertical). Each iteration of automaton evolution is associated with one letter of ${H,V}$.

In this work we focus on ``particles'' rules. That is to say at each site there are a finite amount of particles or chips, that, according to a specific rule are disseminated in the vicinity of a site. Every step is going synchronously, so each site lose and  receive chips simultaneously. In this context the set of states is $\mathbf{Q}={0,1,2, ...} \subset \mathbb{N}$, the number of particles. 

\section{The chip firing automata}
\label{The chip firing automata}

The chip firing automaton, also know as the sandpile model \cite{GMPHYSC,GMTCS,Moore}, is a particular case of the above described particle automata $\langle \mathbb{Z}^2, V, \mathbf{Q}, F, w \rangle$, with the following local function.  
If a site $v\in \mathbb{Z}^2 $ has $x_v \geq 4$ chips then:
\begin{equation}
\begin{split}
  &  x'_v = x_v -4 \\
  &  \forall u \in V_v  \Rightarrow x'_u = x_u +1
\end{split}
\end{equation}
where $V_v$ is the von Neumann neihborhood of the site $v$.

 By adding the condition of open or close sides of the site proposed here; the rule changes as follows:

\begin{equation}
\begin{split}
  &  x_v \geq 4 \Rightarrow x'_v = x_v - \alpha \\
  &  \forall u \in V_v \text{  such that the gate is open} \\
  & \Rightarrow x'_u = x_u +1
\end{split}
\end{equation}

\noindent where $\alpha$ is the number of gates that are open.

When the rule is applied in parallel on every site, the new state at a site $v$ is:

$$x'_v = x_v - \alpha  + \beta$$

\noindent where $\beta$ is the number of chips which the site $v$ receives from its open and  firing neighbours. 

If every side (columns and  rows) is always open, then we have the usual chip firing automaton. When a  word $w$ of open or close sides is  considered, for instance $w= HHVV$, at each step we open (or close) the rows or columns of the grid periodically 

\begin{equation}
    (HHVV)^* = HHVV \text{  } HHVV \text{  } HHVV \text{ ... }
\end{equation}

\section{Computational Complexity notions}
\label{Computational Complexity notions}

In order to study the complexity of an automaton we can analyse a power of the automaton to simulate Boolean functions, i.e., by selecting specific initial conditions and sites as inputs and  outputs to  determine the different Boolean functions the automaton may compute by its dynamics~\cite{BookComp}. More Boolean functions are founded, more complex is the automaton. A similar notion, related to some prediction problems, appears in the framework of the theory of computational complexity. Essentially this is similar to trying to determine the computational time related to the size of a problem, that a Turing Machine take to solve it. In our context, let us consider the following decision or prediction problem.

 \textbf{PRE}: Consider the chip firing fungal automaton $FA= \langle \mathbb{Z}^2, V, \mathbf{Q}, F, w \rangle$, an initial assignment of chips to every site , $x(0)\in \mathbf{Q}^{\mathbb{Z}^2}$, an integer number $T>0$ of steps of the automaton and  a site 
 $v\in \mathbb{Z}^2 $, such that $x_v(0)=0$. Question: Will $x_v(t)>0$ be for some $t\leq T$?\\
 
Of course one may give an answer by running the automaton at most $T$ steps, which can be done in a serial computer in polynomial time. But the  question is a little more tricky: could we answer faster than the serial algorithm, ideally, exponentially faster in polylogarithmic time in a parallel computer with a polynomial number of processor? 

To answer the questions  we consider two classes of decision problems, those belonging to $P$, the class of problems solved by a polynomial algorithm, and the class $NC$, being the problems solved in a parallel computer in $O(log^q n)$ steps (polylogarithmic time). This is straightforward that $NC\subseteq P$ because any parallel algorithm solved in polylogarithmic time  can be simulated efficiently in a serial computer. But the strict inclusion is a very hard open problem (like the well know $P=NP$). 

An other notion from computational complexity related with the possibility that the two classes melt is P-completeness. A problem is \textit{P-complete} if it is in the class $P$, (that is to say, there exist a polynomial algorithm to solve it) and  every other problem in $P$ can be reduced, by a polynomial transformation, to it. Clearly, if one of those  \textit{P-complete} problems is in $NC$, both classes collapsed in one. So, to prove that a problem is $P-Complete$ gives us an idea of its complexity. 

One well known \textit{P-complete} problem is the \textit{Circuit Value Problem}, i.e., the evaluation of a Boolean Circuit (Boolean function). Roughly because any polynomial problem solved in a serial computer (a Turing machine) can be represented as a Boolean circuit. On the other hand, Boolean circuits intuitively are essentially serial because in order to compute a layer of functions it is necessary to compute previous layers so in principle it is not clear how to determine the output of the circuit in parallel. Further, when the circuit is monotonous, i.e., it admits only {\sc or} and {\sc and} gates (no negations) it is also a \textit{P-complete} problem. This is because negation can be put in the input (the two bits of the variable 0 and  1) and  for the gates which are a negation, to use the De Morgan laws.

The complexity of the chip firing automata was first studied in \cite{GMPHYSC,GMTCS}, where it was proved that in arbitrary graphs (in particular, non-planar ones) the chip firing automata are Turing Universal. To prove this  a universal set of Boolean circuits is built by using specific automata configurations, so, also \textbf{PRE} is \textit{P-complete}. In  a similar way, but in a d-dimensional grid, $\mathbb{Z}^d$, it was proved in \cite{Moore} that for $d\geq 3$ the problem is \textit{P-complete} and  the complexity, until today, remains open for a two-dimensional grid. In \cite{gajardo2006crossing} it was proved that in a two-dimensional grid  and  the von Neumann and  Moore neighbourhood it is not possible to cross signals by constructing wires over quiescent configurations. That can be done only for bigger neighbourhood, so, in fact, over non planar graphs.  

\section{Computational Complexity of the Fungal Automata}
\label{Computational Complexity of the Fungal Automata}

We will study the computational complexity of the Fungal Sand Pile Automaton, by proving that for the word $w = V^4H^4= HHHHVVVV$, the Prediction Problem,  \textbf{PRE},  is \textit{P-complete}. That is to say one can not determine an exponentially faster algorithm to answer unless $NC=P$.

\begin{prop}
 For the  word $H^4V^4 = HHHHVVVV $ the fungal chip firing automaton is \textit{P-complete}.
\end{prop}

\begin{proof} 
Clearly the problem is in $P$. It suffices to run the automaton at most $T$ steps and  see if the site $i$ changes, which is done in $O(T^3)$: in fact we have to compute the ``cone" between site $i$, its neighbourhood at step $T-1$, $T-2$, and  so on, until the initial values in the site in an square $(2T-1)\times (2T-1)$. So the number of sites is  to consider is $1^2+3^2+5^2+ ...+2T-1)^2$ which is bounded  by a cubic polynomial, so one may compute the state of site $i$ in $T^3$.

To establish the completeness, we will reduce the Monotone Circuit Value Problem to the Fungal Automaton Prediction Problem, \textbf{PRE}. That is to say, to establish specific automaton configurations which simulates  a wire, the {\sc and}  and the {\sc or} gates, as well as a {\sc cross-over}. This last gadget is important to compute non-planar circuits in the two dimensional grid. 

In the constructions below every cell which is not in the diagram is understood initially empty, without chips ($x_i=0$). 

\begin{figure}[!tbp]
    \centering
    \subfigure[]{\includegraphics[width=0.55\textwidth]{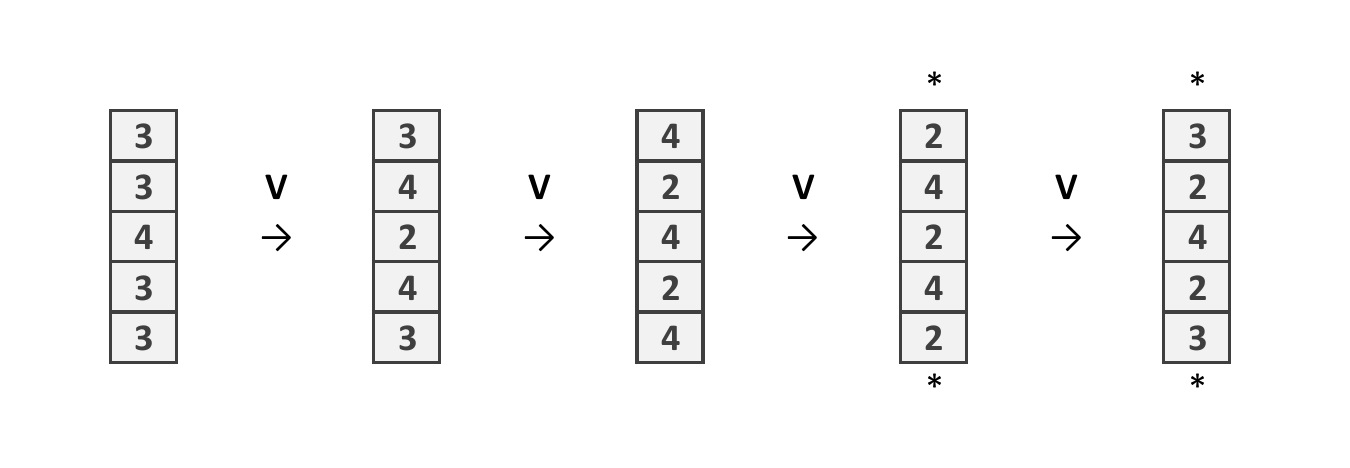}}
    \subfigure[]{ \includegraphics[width=0.8\textwidth]{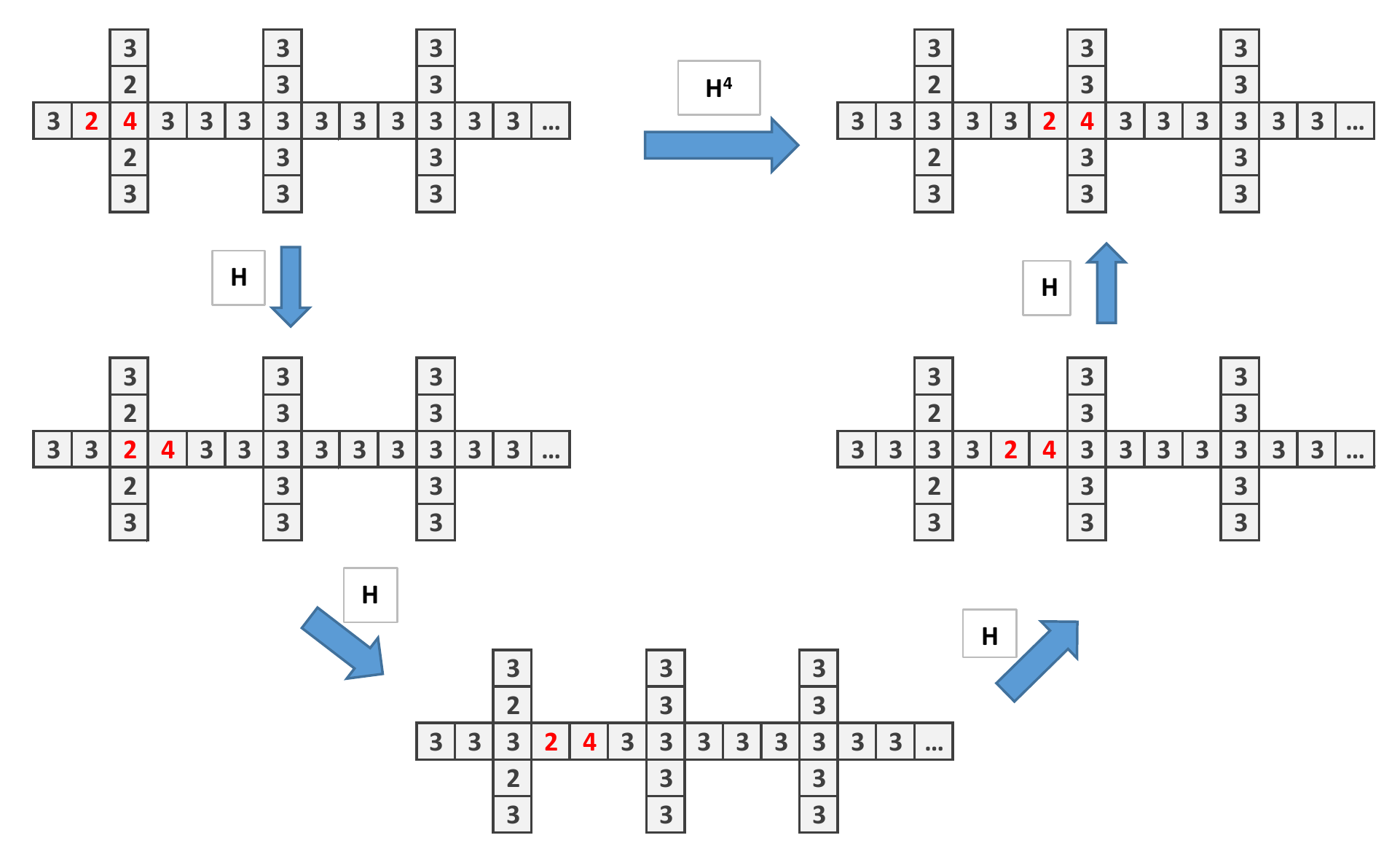}}
     \subfigure[]{\includegraphics[width=0.8\textwidth]{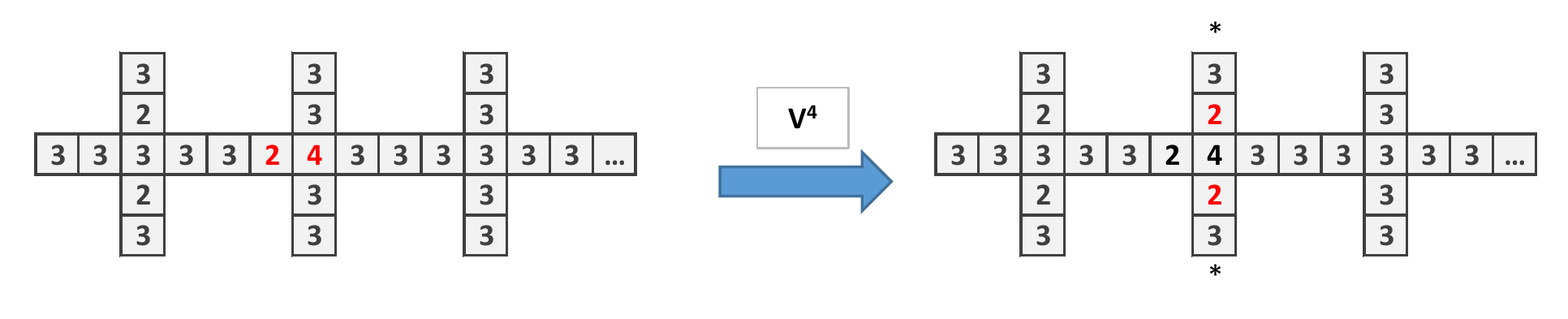}}
    \caption{Wires. 
    (a)~Application of $V^4$ to the wire.
    (b)~Application of $H^4$ to the wire.
    (c)~Final state of the wire after the application of $V^4$.
    }
    \label{fig:wires}
\end{figure}

To construct the wire let us first see what happen when $V^4$ is applied in the particular structure showed in (Fig. \ref{fig:wires}a). The important issue is that the initial site with 4 chips, after the application of $V^4$, remains with 4 chips. Only the adjacent sites and  down change their number of chips.Then one applies $H^4$ to obtain Fig.~\ref{fig:wires}b which is similar to  the initial configuration, shifted to the right  (Fig.~\ref{fig:wires}c). 

\begin{figure}[!tbp]
    \centering
    \includegraphics[width=0.9\textwidth]{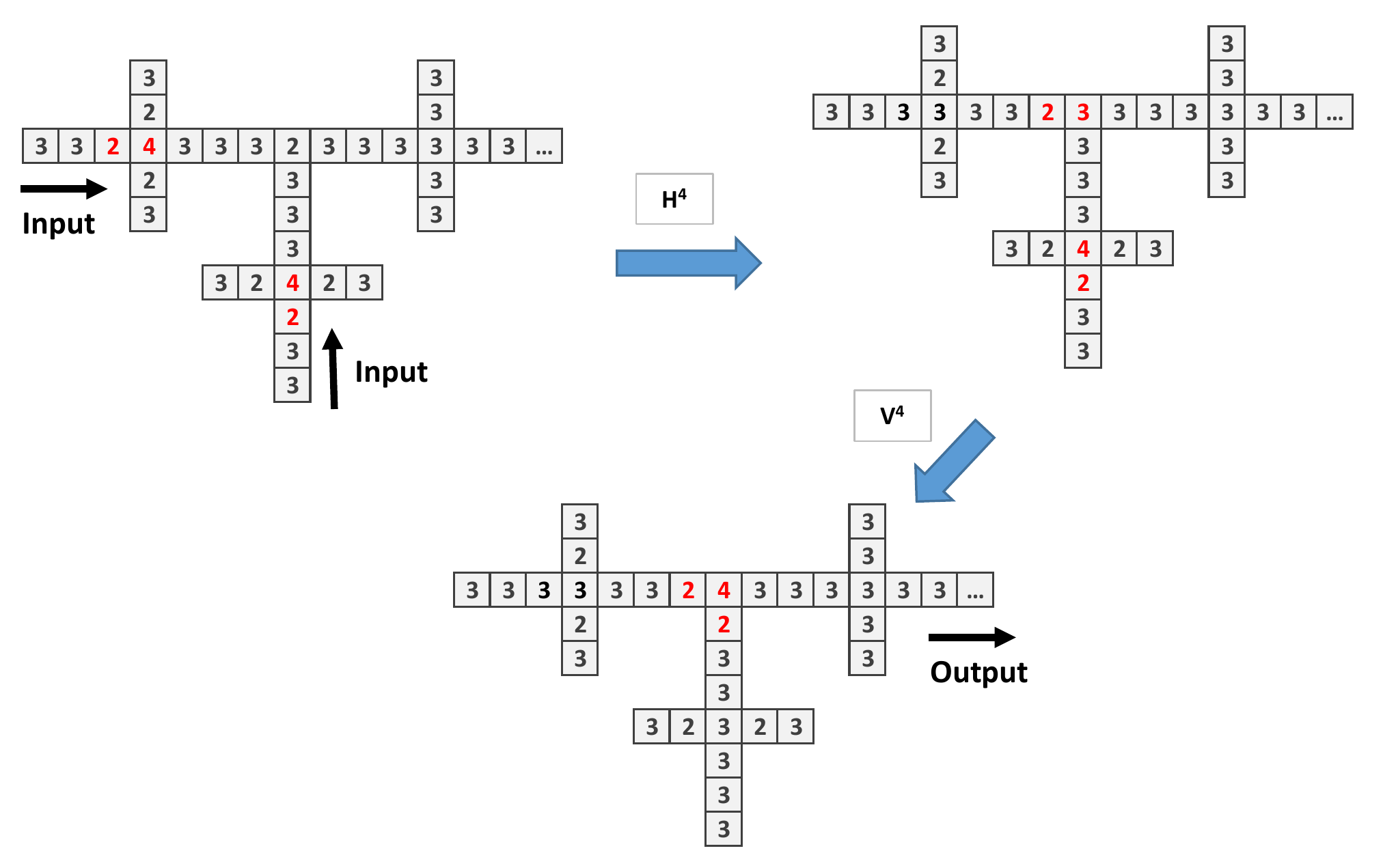}
    \caption{The {\sc and}  gate for inputs with 4 chips (signal= 1).}
    \label{fig:AND1}
\end{figure}

\begin{figure}[!tbp]
    \centering
    \includegraphics[width=0.9\textwidth]{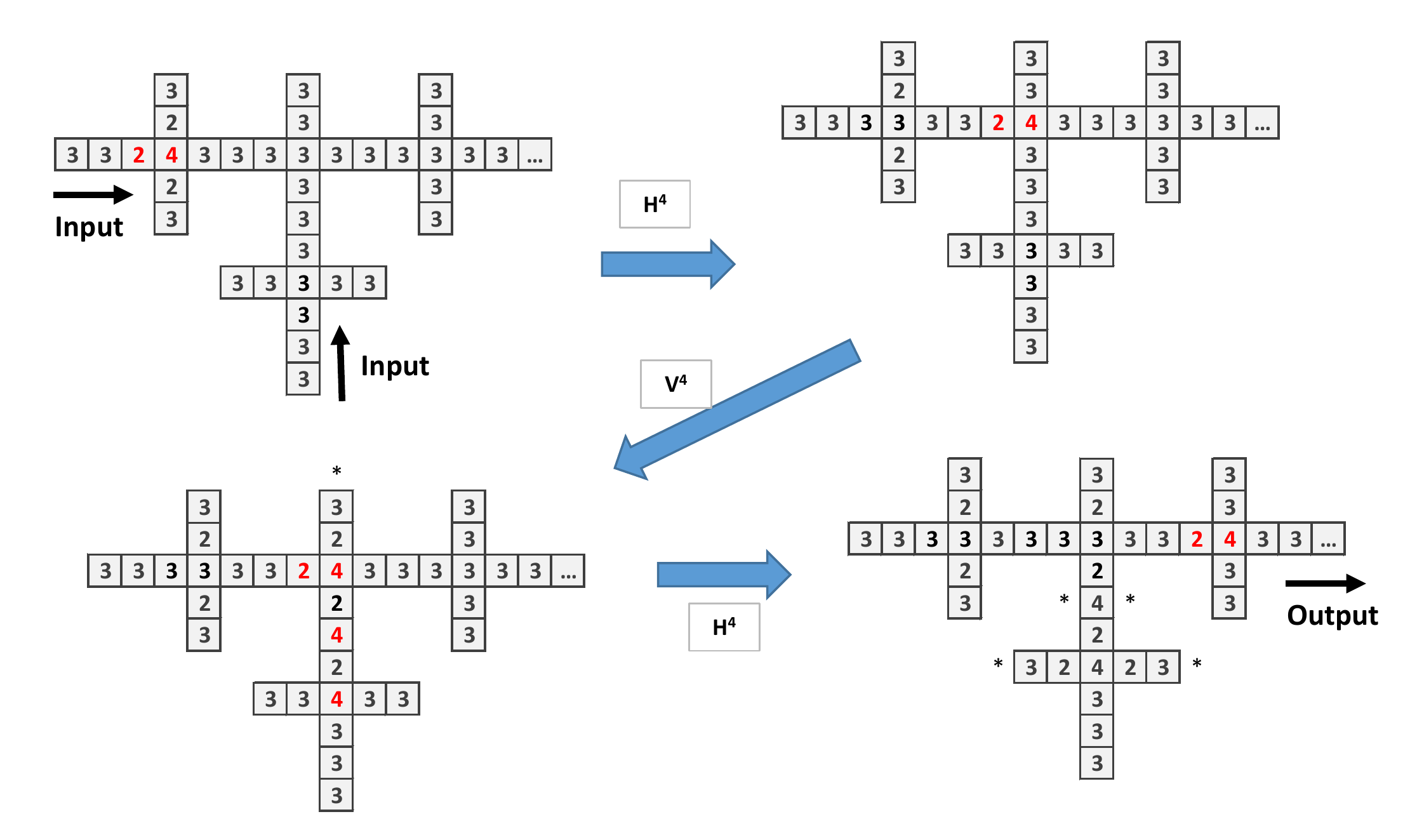}
    \caption{The {\sc or} gate for one firing input.}
    \label{fig:OR1}
\end{figure}

To implement {\sc and} gate  (Fig.~\ref{fig:AND1}) and {\sc or} gate (Fig. \ref{fig:OR1}) we have to connect two wires (this corresponds to  a branching of mycelium).  In the {\sc and} gate two single chips arrive to a central site with 2 chips, so to trigger firing, threshold 3, the signal has to arrive. With $H$ the signal continues to the right, thus the output is 1. The {\sc or} gate functions similarly to {\sc and} gate but the central site has 3 chips. There is an unwanted signal coming back signal  but the computation is made to the right. 

\begin{figure}[!tbp]
    \centering
    \subfigure[]{\includegraphics[width=0.75\textwidth]{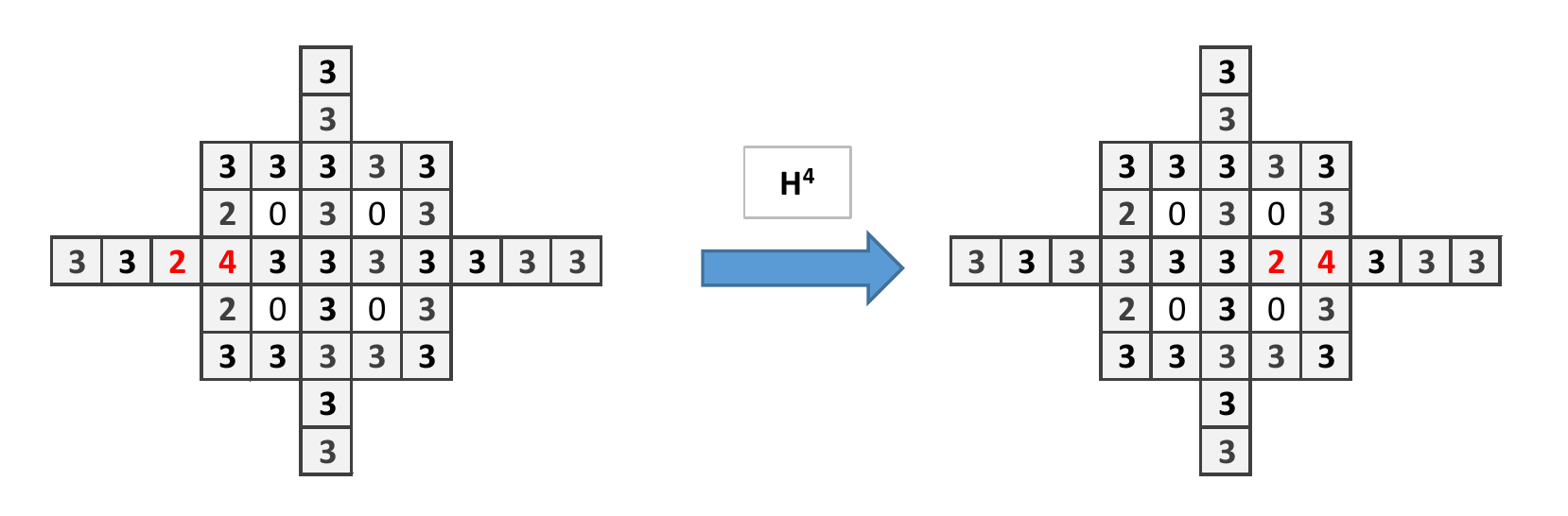}}
    \subfigure[]{\includegraphics[width=0.75\textwidth]{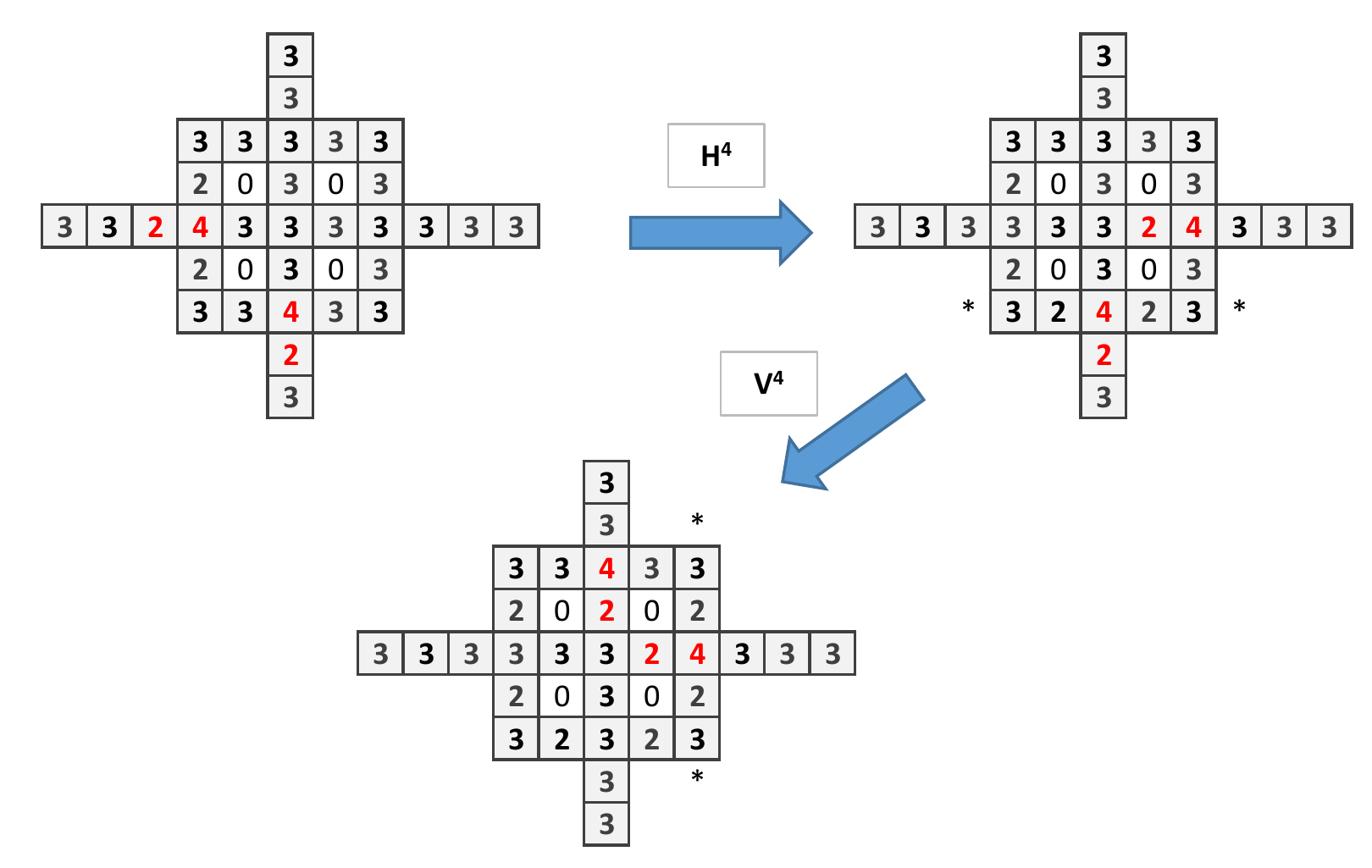}}
    \caption{The cross-over (a) with a horizontal signal and (b) for two signals.}
    \label{fig:crossover}
\end{figure}

The cross-over is demonstrated in Fig.~\ref{fig:crossover}. Here we apply four $V$ and $H$ steps. In Fig. \ref{fig:crossover}a we illustrate the crossing of a horizontal signal (by applying $H^4$). For the vertical signal the dynamics is similar but $V^4$ is applied. Figure~\ref{fig:crossover}b shows the case when two signals arrive at the junction at the same time.
\end{proof}

 \section{Other words of automaton updates}
  \label{Other words of automaton updates}
 
For other shorter words, like the usual chip firing (with sides always open)  and  the words in the set $\mathbf{B}= \{ HV, H^2V^2, H^3V^3\}$ we are able to construct wires, the {\sc or} and  the {\sc and}  gates, but we are unable to built a cross-over. In such cases we can only implement planar circuits with non-crossing wires. 

Below we exhibit the different constructions for those words. It is important to point out that the strategy we used to built the constructions has been by taking as initial framework a quiescent configuration, i.e. a fixed point of the automaton. In \cite{gajardo2006crossing} it has been proved that with this strategy, for a two dimensional grid with the von Neumann or Moore neighbourhoods it is impossible to cross information, i.e. to built a cross-over. It seems that is also the case for the words in the set $\mathbf{B}$. In this sense one may say that our result is the best possible: no shortest word allows to cross information, at least following the quiescent strategy. 

\begin{figure}
    \centering
    \includegraphics[width=0.7\textwidth]{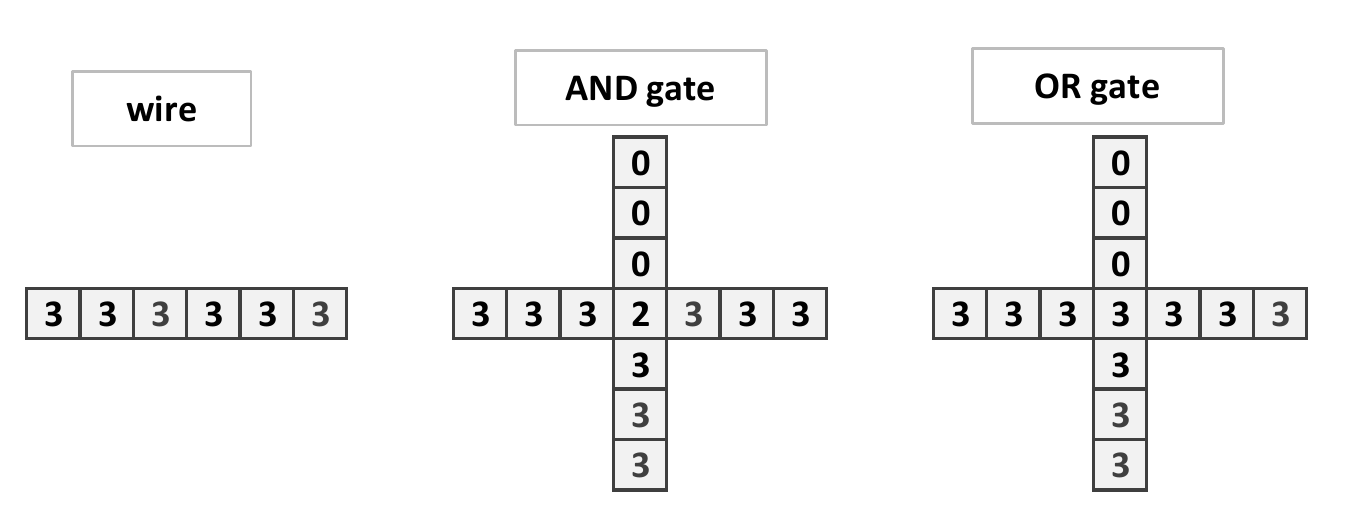}
    \caption{The wire and  gates for the classical chip firing automaton: every side is open.}
    \label{fig:wireAO}
\end{figure}

\begin{figure}[!tbp]
    \centering
    \subfigure[]{\includegraphics[width=0.8\textwidth]{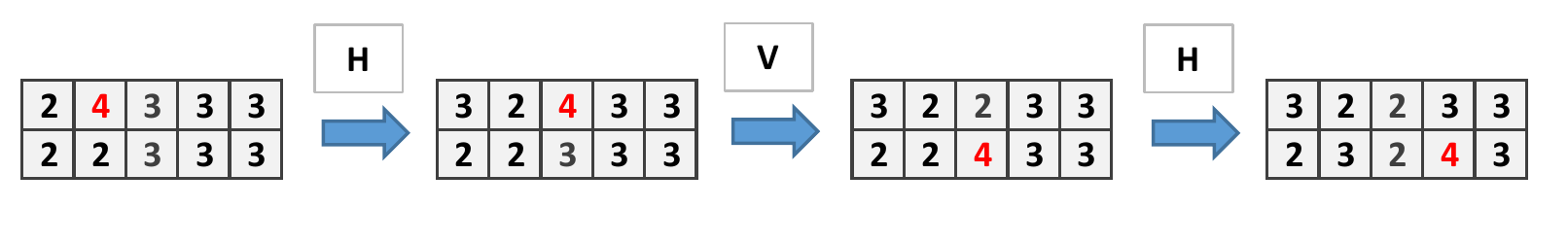}}
     \subfigure[]{ \includegraphics[width=0.75\textwidth]{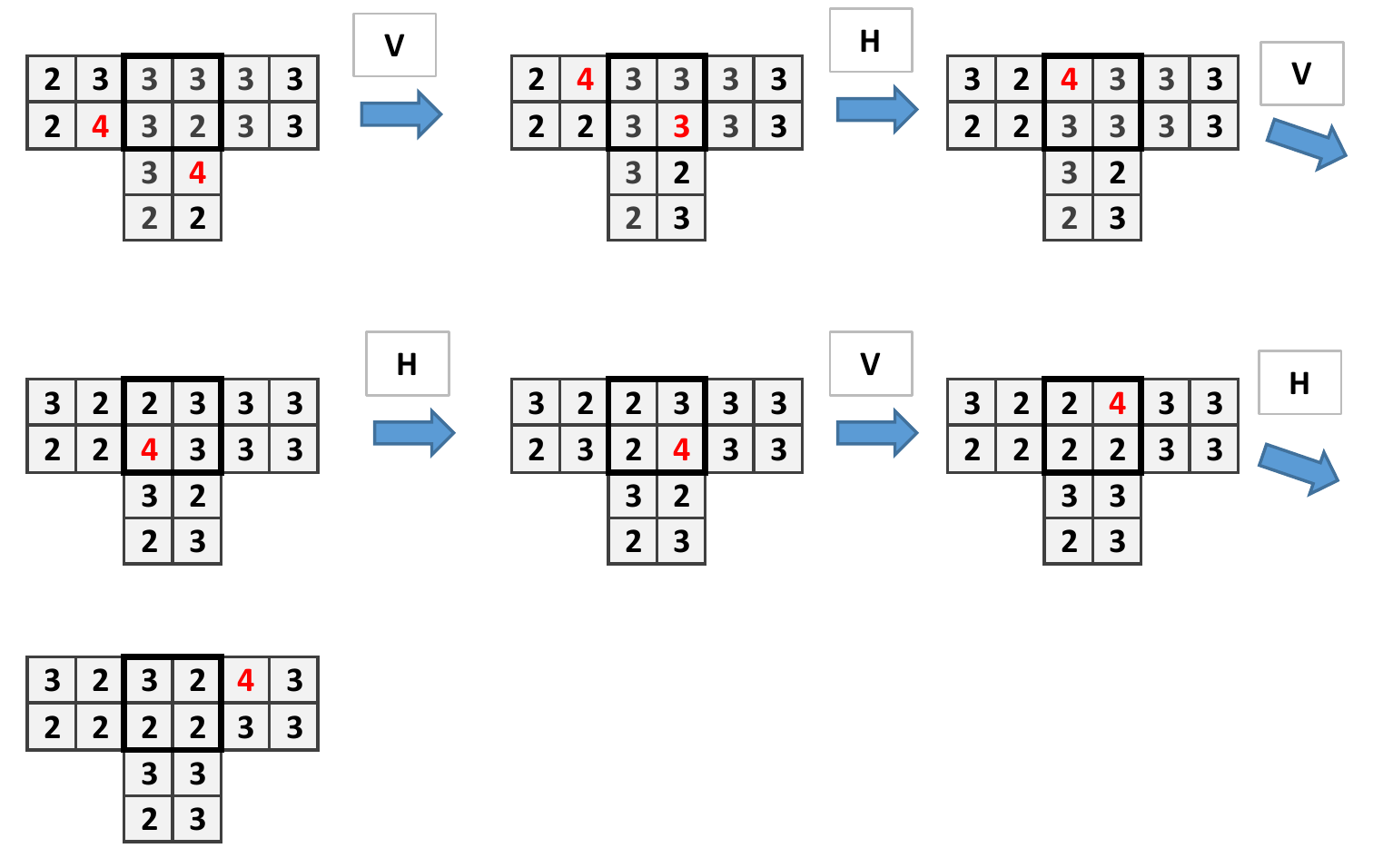}}
    \caption{Operation of the $(HV)^*$ word. 
    (a)~The wire.
    (b)~The {\sc or} gate. 
    }
    \label{fig:HV}
\end{figure}

\begin{figure}[!tbp]
    \centering
    \subfigure[]{\includegraphics[width=0.25\textwidth]{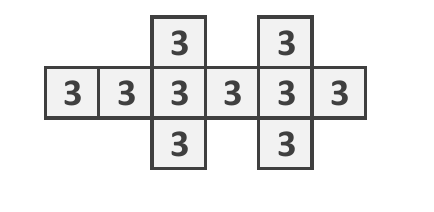}}
     \subfigure[]{\includegraphics[width=0.6\textwidth]{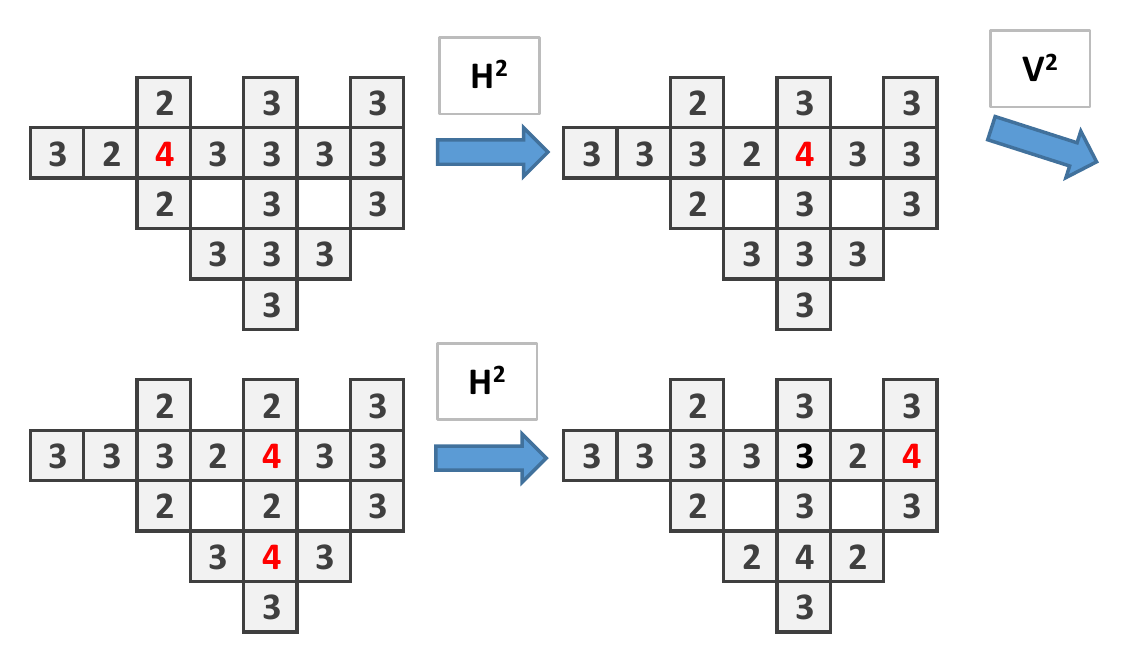}}
      \subfigure[]{\includegraphics[width=0.6\textwidth]{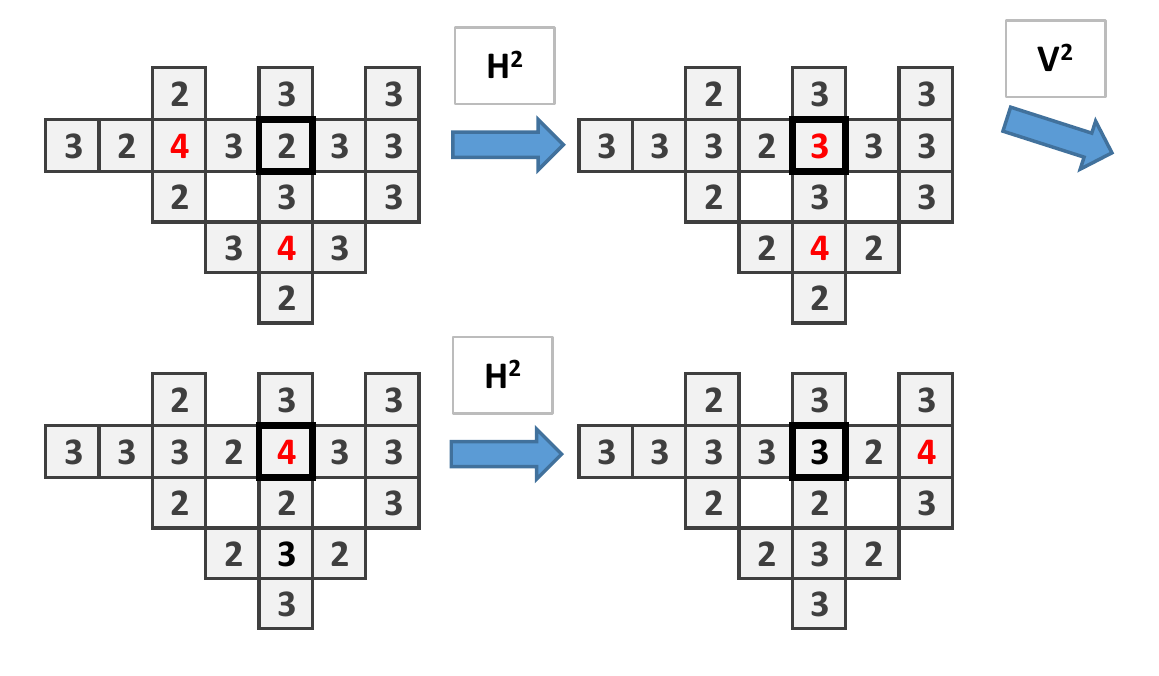}}
    \caption{Operation of the $(HHVV)^*$ word. 
    (a)~The wires.
    (b)~The {\sc or} gate. 
    (c)~The {\sc and} gate.
    }
    \label{fig:HHVV}
\end{figure}

In previous situations for the usual chip firing automaton the constructions are given in Fig.~\ref{fig:wireAO}.  For the word $HV$:  wire is shown in Fig.~\ref{fig:HV}a and the {\sc and} gate in 
Fig.~\ref{fig:HV}.
For the word $HHVV$, the wire is shown in Fig.~\ref{fig:HHVV}a, the {\sc or} gate is shown in Fig.~\ref{fig:HHVV}a  and the {\sc and} gate in Fig.\ref{fig:HHVV}c. 

\section{Discussion} 
\label{Discussion}

Using sandpile, or chip firing, automata we proved that Fungal Automata are computationally universal, i.e., by arranging positions of branching in mycelium it is possible to calculate any Boolean function.

\begin{figure}[!tbp]
    \centering
    \includegraphics[width=0.6\textwidth]{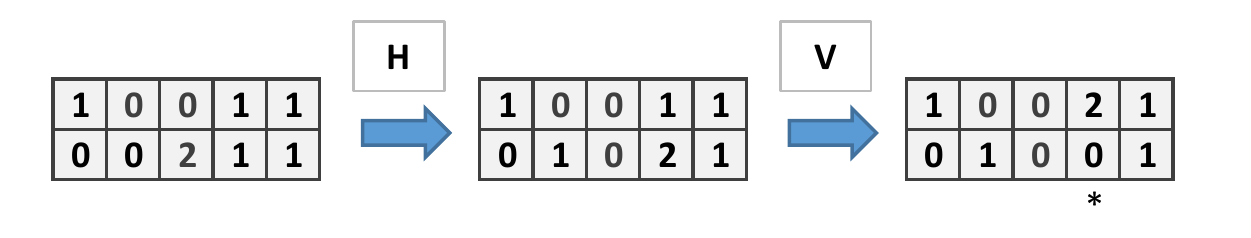}
    \caption{Two chip-firing wire.}
    \label{fig:2chipwire}
\end{figure}

\begin{figure}[!tbp]
    \centering
    \includegraphics[width=0.7\textwidth]{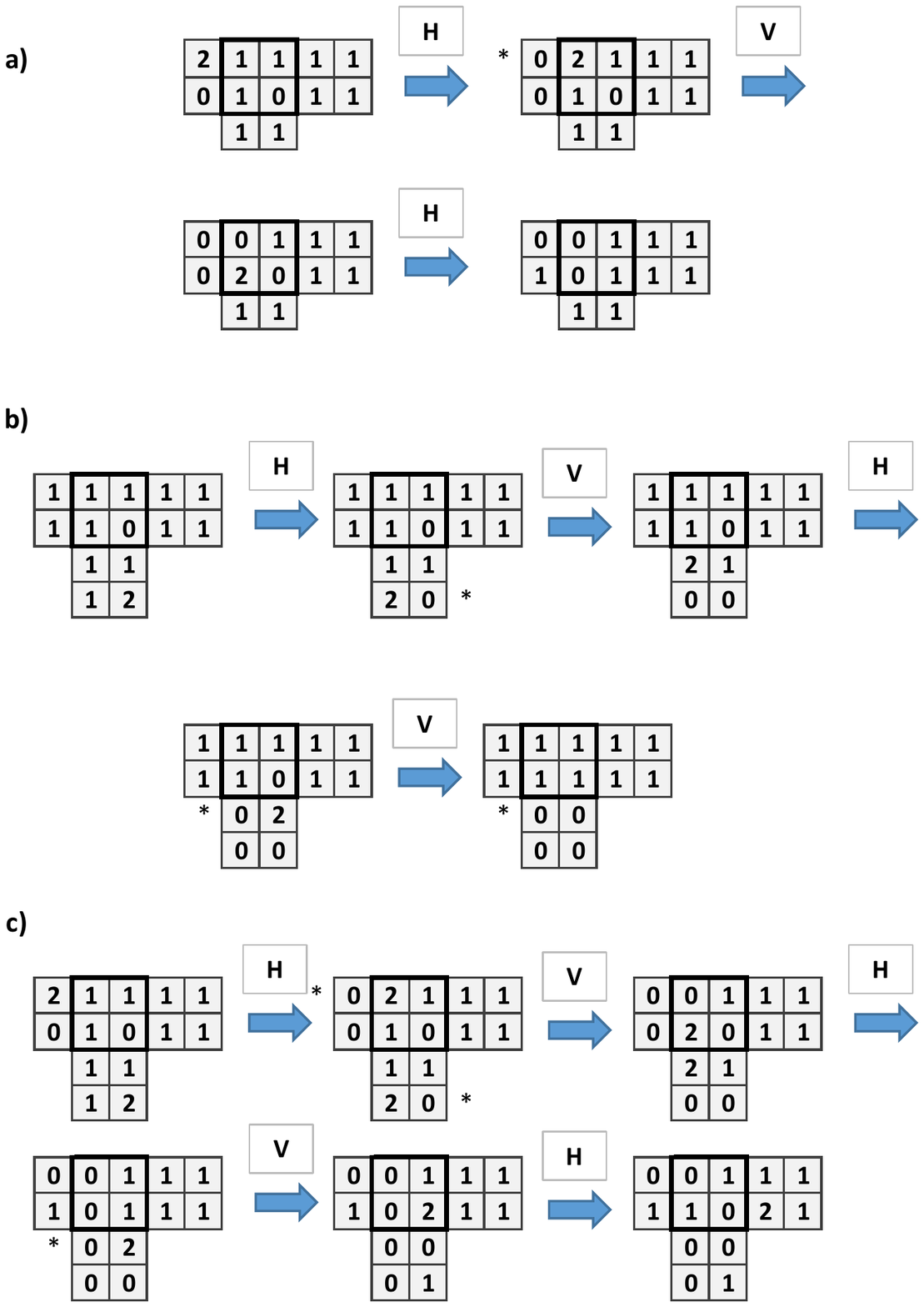}
    \caption{The {\sc and} gate for firing threshold 2. 
    (a)~One horizontal signal. 
    (b)~One vertical signals. 
    (c)~Two signals and  one output signal. }
    \label{fig:2chipAND1}
\end{figure}

The structure of Fungal Automata presented can be relaxed by consider the site firing chips only when it has as many chips as open side. In present model, since at each step there are only two sides can be open, the firing threshold is 2. In this situation, the wire {\sc and} and the {\sc or}  gates can be built as in previous cases but not the cross-over. Dynamics of the wire is shown in Fig.~\ref{fig:2chipwire}, the {\sc and} gate in Fig.~\ref{fig:2chipAND1}.\\

Consider the $N\times N$ grid  $\{0,1,2, ...,N-1\}\times \{0,1,2, ...,N-1\}$. Another possibility to open-close sides could be the following: 
at even steps $t=0,2,...$, open the even rows and  columns and at odd steps,  $t=1,3,5,...$ open the odd  rows and  columns.

If we do that we simulate exactly the Margolous partitions ($2\times 2$ blocks)~\cite{margolus1984physics}. This  give us another way to determine the universality and, in this case, reversibility of this specific Fungal Automaton, because with this strategy one may simulate the Margolous billiard~\cite{margolus2002universal}. Not only that, given any other block partition automaton, say by $p\times p$ blocks there exist a way to open-close the sides which simulates it~\cite{morita1989computation,durand2000reversible,imai2000computation,durand2001representing}.

A significance of the results presented for future implementations of fungal automata with living fungal colonies in experimental laboratory conditions is the following: in our previous research, see details in~\cite{adamatzky2020boolean}, we used FitzHugh-Nagumo model to imitate propagation of excitation on the mycelium network of a single colony of \emph{Aspergillus niger}. Boolean values are encoded by spikes of extracellular potential. We represented binary inputs by electrical impulses on a pair of selected electrodes and we record responses of the colony from sixteen electrodes. We derived sets of two-inputs-on-output logical gates implementable the fungal colony and analyse distributions of the gates~\cite{adamatzky2020boolean}. Indeed, there were combination of functionally complete sets of gates, thus computing with travelling spikes is universal. However, in~\cite{adamatzky2020boolean}, we made a range of assumptions about origins, mechanisms of propagation and interactions 
of impulses of electrical activity. If the spikes of electrical potential do not actually propagate along the mycelium the model might be incorrect. The sandpile model designed in present paper is more relaxed because does not any auto-catalytic processes: avalanches can be physically simulated by applying constant currents, chemical stimulation to mycelium network. This is because the avalanches can be seen as movement of cytoplasm of products of fungal metabolism.

Whilst thinking about potential experimental implementation initiating avalanches is just one part of the problem. Selective control of the Woronin bodies might bring substantial challenges. As previous studies indicate the Woronin bodies can block the pores due to cytoplasmic flow~\cite{steinberg2017woronin} or mechanical stimulation of the cell wall, high temperatures, carbon and nitrogen starvation, high osmolarity and low pH~\cite{tegelaar2017functional,soundararajan2004woronin,jedd2011fungal}. We are unaware of experimental studies on controlling Woronin bodies with light but we believe this is not impossible.

 \section*{Acknowledgement}

AA, MT, HABW have received funding from the European Union's Horizon 2020 research and innovation programme FET OPEN ``Challenging current thinking'' under grant agreement No 858132. EG residency in UWE has been supported by funding from the Leverhulme Trust under the Visiting Research Professorship grant VP2-2018-001 and  from the project the project 1200006, FONDECYT-Chile. 

\bibliographystyle{plain}
\bibliography{ref,fungalcompbib,1dfungalbib}

\end{document}